\documentclass{article}
\usepackage[utf8]{inputenc}
\usepackage{graphicx}
\usepackage{amssymb}
\usepackage{amsthm}
\usepackage{amsmath}
\usepackage{mathtools}
\usepackage{bm}
\usepackage{esvect}
\usepackage[noend]{algpseudocode}
\usepackage{algorithm}
\usepackage{tikz} 
\usepackage{chemfig}
\usepackage{tikz}
\usepackage{tikz-cd}
\usepackage{mathrsfs}
\usepackage{lipsum}
\usepackage[colorlinks=true, allcolors=blue]{hyperref}

\usepackage[margin=1.25in]{geometry}
\usepackage{authblk}

\theoremstyle{plain}
\title{An Algebraic Approach to The Longest Path Problem}
\author{Omar Al - Khazali, b00088673@aus.edu}
\affil{Department of Mathematics \& Statistics, American University of Sharjah Department of Computer Science, American University of Sharjah}
\date{}

\begin{document}

\maketitle
\newtheorem*{theorem*}{Theorem}
\newtheorem*{remark*}{Remark}
\newtheorem*{remark}{Remark}
\newtheorem*{empty*}{}
\newtheorem*{reference*}{Reference}
\newtheorem{theorem}{Theorem}
\newtheorem{lemma}{Lemma}
\newtheorem{proposition}{Proposition}
\newtheorem{scolium}{Scolium}   
\newtheorem{definition}{Definition}
\newenvironment{AMS}{}{}
\newenvironment{keywords}{}{}
\newtheorem{sublemma}{Lemma}[section]
\newtheorem{corollary}{Corollary}[theorem]

\providecommand{\keywords}[1]{\textbf{\textit{Index terms---}} #1}

\begin{abstract}
    The Longest Path Problem is a question of finding the maximum length between pairs of vertices of a graph. In the general case, the problem is $\mathcal{NP}$-complete. However, there is a small collection of graph classes for which there exists an efficient solution. Current approaches involve either approximation or computational enumeration. For Tree-like classes of graphs, there are approximation and enumeration algorithms which solves the problem efficiently. Despite this, we propose a new method of approaching the longest path problem with exact algebraic solutions that give rise to polynomial-time algorithms. Our method provides algorithms that are proven correct by their underlying algebraic operations unlike existing purely algorithmic solutions to this problem. We introduce a `booleanize' mapping on the adjacency matrix of a graph which we prove identifies the solution for trees, uniform block graphs, block graphs, and directed acyclic graphs with exact conditions and associated polynomial-time algorithms. In addition, we display additional algorithms that can generate every possible longest path of acyclic graphs in efficient time, as well as for block graphs.
\end{abstract}
$\textbf{\textit{Keywords---}}$Longest Path Problem, Algebraic Graph theory, Algorithms, Combinatorics

\section{Introduction}
The Longest Path Problem (\textbf{LPP}) is a well known challenge in combinatorial optimization which is directly linked to the prominent Hamiltonian Path Problem (\textbf{HPP}). The ability to solve \textbf{LPP} in polynomial time implies the existence of efficient solutions of \textbf{HPP}. Specifically, for a graph of order $n$ if there exists a path of length $n - 1$ then there exists a Hamiltonian path. More generally, $\textbf{HPP}$ presents itself as a special case of $\textbf{LPP}$ \cite{Schrijver}. Overall, the aim is to solve \textbf{LPP} on a certain family of graphs in polynomial-complexity, either by estimation to rule out candidate solutions, or to find the exact solution for a given class. The main approaches to solving \textbf{LPP} involve approximation algorithms, or computational enumeration/parametrized complexity.

There are existing efficient and optimal solutions, with more developments on further classes of graphs. W. Dijkstra's (longest path) algorithm works on Trees in $O(n)$ time as was proven by R.W. Bulterman et al \cite{Bulterman}. In addition, a generalization of Dijkstra's algorithm was introduced by R. Uehara and Y. Uno to solve weighted Trees in $O(n)$, interval biconvex graphs in $O(n^3(m + n\log n))$, block graphs in $O(n + m)$, ptolemaic graphs in $O(n(n+m))$, cacti graphs in $O(n^2)$, and threshold graphs in $O(n + m)$ \cite{Uehara}. Moreover,  K. Ioannidou and S. D. Nikolopoulos solved cocomparability graphs in $O(n^7)$ \cite{Ioannidou_co}, and with G. B. Mertzios solved interval graphs in $O(n^4)$ \cite{Ioannidou_in}. Furthermore, Y. Guo, C. Ho, M. Ko discovered a polynomial time algorithm on distance-hereditary graphs in $O(n^4)$ \cite{Guo}. Whereby n and m are the order and size of the graph respectively.

We introduce a separate approach to \textbf{LPP} which largely focuses on algebraic conditions that exactly identify the length of the longest path. In particular, algebraic operations which are computable in polynomial time implies that any polynomial-sequence of these algebraic operations would be guaranteed to be in polynomial time. Hence, this approach is unlikely to achieve optimal solutions. However, it does provide the ability to find polynomial time solutions to larger classes of graphs without the requirement of weight or distance functions, nor a constraint to be defined on strictly undirected graphs. Purely algorithmic methods, like Dijkstra's longest path algorithm, are not straightforward to prove correctness of the algorithm. It took many decades before we have seen R. W. Bulterman et al's proof of correctness for Dijkstra's algorithm \cite{Bulterman}.

The use of algebraic techniques that can then easily form into algorithms provides us with a straightforward technique to find polynomial-time solutions to \textbf{LPP} that are concretely and rigorously proven correct unlike existing methods. To the author's knowledge, there does not exist a purely algebraic solution to \textbf{LPP} for any large class of graphs. In this paper, we present a map eluding to algebraic conditions to find the exact solution to \textbf{LPP} for directed acyclic graphs, trees, uniform block graphs, and block graphs in polynomial complexity. Furthermore, we add a section on path-generating algorithms to find every path of the longest length, which is not possible with purely algorithmic methods. Finally, we analyze the computational complexity of the mentioned algorithms, and provide a final remark.

\section{Preliminaries}
A graph $\Gamma = (\mathbb{V}(\Gamma), \mathbb{E}(\Gamma))$ is a pair with a vertex set $\mathbb{V}$ and edge set $\mathbb{E} \subseteq \mathbb{V} \times \mathbb{V}$ respectively. We take $\Gamma$ to be finite, simple, connected, and undirected unless otherwise stated. A walk, $\omega$, is a sequence $(v_1, v_2, ..., v_n)$ where every two adjacent vertices are neighbors such that $\{v_i, v_{i+1}\} \in \mathbb{E}(\Gamma)$. Note that the vertices of a walk need not be distinct. Similarly, a path $\gamma$ is a sequence of vertices $(v_1, v_2, ..., v_n)$ where every two adjacent vertices in $\gamma$ are adjacent in $\Gamma$ and every $v_i \in \gamma$ is visited only once. For the lengths of walks and paths, we denote them by $\mathfrak{L}(\omega)$ and $\mathscr{L}(\gamma)$ respectively such that the length is the number of edges traversed. Finally, we denote the set of neighboring vertices of a fixed vertex $v$ by $N(v) = \{u \in \mathbb{V}\ |\ \{u, v\} \in \mathbb{E}\}$

For graph classes, we recall that a connected graph is a Tree if and only if there are no cycle subgraphs. Or equivalently, there is a unique path $\gamma_{u,v}$ between every distinct $u,v \in \mathbb{V}$. Furthermore, recall that a Block Graph is a connected graph with vertex-adjacent completely-connected subgraphs (cliques), such that $\exists\ m_i > 2$ so that $B_i \cong K_m\ \forall\ B_i \subseteq \Gamma \textit{ maximal 2-connected}$. Recall that a \textit{maximal 2-connected} subgraph, also known as a biconnected component, is a maximal subgraph so that if any one vertex is cut, then the subgraph stays connected. As a special case, we define a Uniform Block Graph to be a Block Graph where all blocks have the same clique order, in other words $\exists!\ m_i > 2 : B_i \cong K_m\ \forall\ B_i \subseteq \Gamma \textit{ maximal 2-connected}$. Finally, a Directed Acyclic Graph (DAG) is a directed graph, $\mathcal{D}$, with a unique path between every two distinct vertices.

In recollection, the adjacency matrix $A(\Gamma)$ is an $\mathbb{M}_{|\mathbb{V}| \times |\mathbb{V}|}$ matrix where $A(\Gamma)_{i,j}$ is $1$ if $v_i, v_j$ are connected by an edge, and otherwise $0$. We denote the distance between two vertices as $\partial(u,v)$ which is the length of the smallest path between $u, v \in \mathbb{V}$. Finally, $\omega(\Gamma)$ is defined to be the order of the largest completely connected subgraph of $\Gamma$. As a final note of clarificaiton, we denote $\mathbb{Z}_n = \{0, 1, 2, ..., n-1\}$ and $\mathbb{Z}_n^{*} = \mathbb{Z}_n \setminus \{0\}$. Now we define a map to ``booleanize" a matrix to a (0,1)-matrix by the following.
\begin{definition}
    Define a map, $\beta$, to `booleanize' any matrix to a (0,1)-matrix as follows
$$\beta : \mathbb{M}_{n \times n} \longrightarrow \mathbb{M}_{n \times n}^{(0,1)}$$
\[
    \beta(Q)_{i,j} := \begin{cases} 1,& Q_{i,j} \neq 0\\
    0,&  Q_{i,j} = 0\\ \end{cases}
\]
reducing non-zero valued entries to 1, and zero valued entries kept as 0 to transform the matrix $A$ to a binary-entry matrix $\beta(Q)_{i,j} \in \{0,1\} \ \forall\ i,j \in \mathbb{Z}_{n}$
\end{definition}
\noindent The interest is to use this transformation on results of multiplicative matrix operations that are required for our algebraic solution of \textbf{LPP}, leading to the definition of longest path length-finding algorithms. Currently, the matrix multiplication algorithm with the optimal complexity is in $O(n^{2.3728596})$. However, given that the algorithm is galactic, in realistic implementations, Strassen's algorithm would be used which is in $O(n^{log_27}) \simeq O(n^{2.807})$, but we assume optimal asymptotic complexity \cite{Duan}. With the goal of minimizing the complexity of the booleanization operation, the method of computation is detailed below for the coming algorithms. Hence, for the purposes of the problem tackled in this paper, we introduce a faster method of computing $\beta(A^2)$ by the following derivation.
\begin{lemma}
    Let $\mathbb{Z}_2^{n}$ be binary numbers with a logical-and operation $\land_l$ such that $(a \land_l b)_i := a_i \land_l b_i$. Let $f$ be a natural map from a matrix row or column vector (of 0's and 1's) to $\mathbb{Z}_2^{n}$. Then, the booleanized product can be computed equivalently by
    $$\beta(A(\Gamma)^2)_{i,j} =  
\begin{cases} 1,& f(A(\Gamma)_{i,-}) \land_l f(A(\Gamma)_{-,j}) > 0\\
    0,&  f(A(\Gamma)_{i,-}) \land_l f(A(\Gamma)_{-,j}) = 0\\ \end{cases}$$
\end{lemma}

\begin{proof}
Let $a, b \in \mathbb{Z}_2^{n}$ and $X, Y \in \mathbb{M}_{n \times n}^{(0,1)}$, we see that
$$(a \land_l b) = 0 \iff a_k = 0 \text{ or } b_k = 0\ \forall\ k \in \mathbb{Z}_{n+1}^{*}$$
$$(a \land_l b) > 0 \iff \exists\ k \in \mathbb{Z}_{n+1}^{*} \text{ so that } a_k = 1 \text{ and } b_k = 1$$
Therefore, modelling it in terms of booleanized matrices, we have the same structure
\begin{align*}
    \beta(X \times Y)_{i,j} &= min\left\{1, \sum_{k=1}^nX_{i,k} \cdot Y_{k,j}\right\}\\
    \beta(X \times Y)_{i,j} = 0 &\iff X_{i,k} = 0 \text{ or } Y_{k,j} = 0\ \forall\ k \in \mathbb{Z}_{n+1}^*\\
    \beta(X \times Y)_{i,j} = 1 &\iff \exists\ k \in \mathbb{Z}_{n+1}^* \text{ so that } X_{i,k} > 0 \text{ and } Y_{k,j} > 0
\end{align*}
Since adjacency matrices are of (0,1)-entries, then we can map a row (or column) of the matrix to a binary number by the natural map, $f$, meaning there are binary numbers $a, b \in \mathbb{Z}_2^{n}$ so that $f(A(\Gamma)_{i,-}) = a$ and $ f(A(\Gamma)_{-,j}) = b$. Where  $A(\Gamma)_{i,-}$ and $ A(\Gamma)_{-,j}$ are the $i^{th}$ Row and $j^{th}$ Column vectors respectively. So we can thake the operation of multiplying a row and column of a matrix as the logical-and operation in binary numbers (after booleanizing). Noting that $f$ is a natural map taking 0 to 0 and any other number to 1 whereby $A(\Gamma)_{i,k}$ maps to $a_k$ and $A(\Gamma)_{k,j}$ maps to $b_k$. Hence we deduce
\[
\beta(A(\Gamma)^2)_{i,j} = min \left\{1, \sum_{k=1}^nA(\Gamma)_{i,k} \cdot A(\Gamma)_{k,j}\right\} = 
\begin{cases} 1,& f(A(\Gamma)_{i,-}) \land_l f(A(\Gamma)_{-,j}) > 0\\
    0,&  f(A(\Gamma)_{i,-}) \land_l f(A(\Gamma)_{-,j}) = 0\\ \end{cases}
\]
\end{proof}
Given that bitwise operations are computable in $O(1)$ or some $O(\delta)$ where $\delta$ depends on the number of bytes of an integer. That reduces the complexity of computing the booleanization of a matrix product to $O(\delta n^{2})$. Allowing us to reduce the complexity of the algorithms pertaining to the length of the longest path. Specifically referring to algorithms other than Algorithm 1, which will be discussed in Section 4 and 5. While Algorithm 1 defines the process of using the above method to compute a booleanized matrix product.

The motivation for this is to take the adjacency matrix of a given graph, $A(\Gamma)$ and compute its powers $A(\Gamma)^n$ for integers $n$ where we find some condition or point at which we can conclude what the longest-path's length is. The alternative to this method would be to compute the result of the matrix product $A^k$ then booleanizing it with $\beta(A^k)$ which would be of complexity $O(n^{2.807}\cdot log_2(k))$ assuming Strassen's algorithm and using the binary exponentiation technique. However, we will see that this method would yield to a faster algorithm for the same result. Below, we define the algorithm to compute $\beta(A(\Gamma_1) \times A(\Gamma_2))$, assuming that they are of the same order.

\newpage

\begin{algorithm}
\caption{Booleanization Algorithm \& Computing $\beta(A(\Gamma_1) \times A(\Gamma_2))$}\label{euclid}
\begin{algorithmic}[1]
\Require $\Gamma_1 := (\mathbb{V}(\Gamma), \mathbb{E}(\Gamma_1))$
\Require $\Gamma_2 := (\mathbb{V}(\Gamma), \mathbb{E}(\Gamma_2))$
\For{$i \gets 1, ..., |\mathbb{V}(\Gamma)|$}
\For{$j \gets 1, ..., |\mathbb{V}(\Gamma)|$}
\If{$A(\Gamma_1)_{i,j} = 1 \land A(\Gamma_2)_{j,i} = 1$}
\State $A(\Gamma')_{i,-} \gets A(\Gamma')_{i,-} + 2^j$
\State $A(\Gamma')_{-,j} \gets A(\Gamma')_{-,j} + 2^i$
\EndIf
\EndFor
\EndFor
\For{$i \gets 1, ..., |\mathbb{V}(\Gamma)|$}
\For{$j \gets 1, ..., |\mathbb{V}(\Gamma)|$}
\If {$A(\Gamma)_{i,-} \land A(\Gamma)_{-,j} > 0$}
\State $R_{i,j} \gets 1$
\EndIf
\EndFor
\EndFor
\State $\textbf{return} \ R$
\end{algorithmic}
\end{algorithm}

Computing the complexity of this algorithm is straightforward and defined in Section 6, Corollary 4.2. We find that it is of complexity $O(\delta n^2)$ for some small $\delta$ dependant on the choice of interpreting the computation of logical-and.

\begin{lemma}[N. Biggs \cite{Biggs}] The number of walks from $v_i$ to $v_j$ can be found by
$$|\{\omega = (v_i, ...,v_j)\ |\ \mathfrak{L}(\omega) = k : v_i, v_j \in \mathbb{V}(\Gamma)\}| = A(\Gamma)^k_{i,j}$$
\end{lemma}

Using Lemma 2 with the booleanization map leads to us having a description on whether a path of length $n$ exists between two vertices $v_i, v_j \in \mathbb{V}$. In particular, Lemma 3 would state exactly when we can find a path of length $n$. However, first, a remark
\begin{remark}
    If there exists a walk $\omega = (v_i, ..., v_j)$ then there exists a path $\gamma = (v_i, ..., v_j)$
\end{remark}
Which is clear enough as is. Since, if there exists a sequence of adjacent vertices from $v_i$ to $v_j$. Then, there must exist a minimal-walk with the least number of elements that will be equal to a path. We use this fact for all our results, and leads to another clear remark that comes together as a consequence of this fact.
\begin{remark}
    If $A(\Gamma)^k_{i,j} \neq 0$ \textnormal{(}i.e. $\beta(A(\Gamma)^{k})_{i,j} = 1$\textnormal{)}. Then, there must be a path from $v_i$ to $v_j$
\end{remark}
Therefore, the idea behind booleanizing the adjacency matrix taken to a certain power is to precisely see if there even exists a path from $v_i$ to $v_j$. We can then use this to evaluate what the length of the longest path would be. Lemma 3 specifies this fact further for Trees where we bound the length of the walk to only certain possible lengths for a given path.
\begin{lemma}
Let $\Gamma$ be a simple, undirected, finite tree. Then, 
$$\beta(A(\Gamma)^n)_{i,j} = 1 \iff \exists\ \gamma = (v_i, ..., v_j) \text{ and } n \in \{\mathscr{L}(\gamma) + 2k \ | \ k \in \mathbb{N}\}$$
\end{lemma}
\begin{proof}$\implies$\newline
    $\Gamma$ acyclic implies that $girth(\Gamma) > 3$ meaning that $A(\Gamma)^n_{i,j} > 0$ iff there is a walk $\omega = (v_i, ..., v_j)$
    that has a length $\mathfrak{L}(\omega) = \partial(v_i,v_j) + 2k : k \in \mathbb{N}$. Given that to go from $v_j \longrightarrow v_j,$
    there must be two steps to loop back $v_j\xrightarrow{} u\xrightarrow{} v_j : u \in N(v_j)$. Therefore, since there is a unique path $\gamma = (u,...,v) \ \forall\ u,v \in \mathbb{V}(\Gamma)$ then,
    $\mathscr{L}(\gamma) = \partial(u,v)$ implies that $\mathfrak{L}(\omega) = n = \mathscr{L}(\gamma) + 2k$
    $\implies A(\Gamma)^n_{i,j} > 0 \iff \beta(A(\Gamma)^n)_{i,j} = 1 \iff n \in \left\{\mathscr{L}(\gamma) + 2k \ | \ k \in \mathbb{N}\right\}$\\
    
\noindent \textit{Proof. }$\Longleftarrow$\newline
    $\exists\ \gamma = (v_i, ..., v_j)$ and $n \in \{\mathscr{L}(\gamma) + 2k \ | \ k \in \mathbb{N}\}$ so there is a walk $\omega = (v_i, ..., v_j)$ such that 
    $\mathfrak{L}(\omega) = \partial(v_i,v_j) + 2k : k \in \mathbb{N}$ given that $girth(\Gamma) > 3$. Therefore since there is a unique path $ \gamma = (u,...,v)$ for every $u, v \in \mathbb{V}(\Gamma)$. Then, $\mathscr{L}(\gamma) = \partial(u,v) \implies \mathfrak{L}(\omega) = n = \mathscr{L}(\gamma) + 2k$
    giving us a walk $\omega = (v_i, ..., v_j)$ of length $n$, $ \mathfrak{L}(\omega) = n$. Therefore, $A(\Gamma)^n_{i,j} \neq 0$ so $\beta(A(\Gamma)^n)_{i,j} = 1$
\end{proof}

\newpage

\section{The Longest Path Problem}
Now that we have a description for when we can find a path of length $n$, we notice a pattern in the booleanized matrix powers where we find an upper limit for the longest paths. By finding the maximum possible power, $n \in \mathbb{N}$, where the matrix $\beta(A(\Gamma)^n)$ ``no longer changes'' it then indicates that we have reached the maximum path length possible. Specifically, we have
\begin{theorem}
Let $\Gamma$ be a Tree, then TFAE:
\begin{enumerate}
    \item The length of the Longest Path of $\Gamma$ is n
    \item $\beta(A(\Gamma)^{n+1}) = \beta(A(\Gamma)^{n-1}) :$ n is the minimum such number
\end{enumerate}
Equivalently, to compute the value,
The Length of the Longest Path is found by the following
\[
    \mathfrak{D}(\Gamma) = min \{n \ |\  \beta(A(\Gamma)^{n+1})  = \beta(A(\Gamma)^{n-1})\}
\]
\end{theorem}

\noindent The logical interpretation is that if some matrix for a booleanized $A(\Gamma)^{n+1}$ is equal to a previously calculated booleanized $A(\Gamma)^{n-1}$, then the minimum $n \in \mathbb{N}_0$ for which that property is true is the longest path. Which follows with an equivalent statement whereby:
\[
    k \ge \mathfrak{D}(\Gamma) \iff \beta(A(\Gamma)^{k+1}) = \beta(A(\Gamma)^{k-1})
\]

\begin{proof}
\subsubsection*{$1 \implies 2$}
Assume $\mathfrak{D}(\Gamma) = n$. Then $\forall\ u,v \in \mathbb{V}(\Gamma)\ \nexists\ \gamma = (u, ..., v)  : \mathscr{L}(\gamma) = n+1$. Hence, by Lemma 2 and 3, we can see that between any pair of vertices we have a walk for any $u, v \in \mathbb{V}(\Gamma) \ \exists\ \omega = (u, ..., v) $ such that $ \mathfrak{L}(\omega) = n+1 \iff \exists\ \Omega \subset \omega $ where $ \mathfrak{L}(\Omega) = n-1$. Since $\beta(A(\Gamma)^k)$ shows the existence of a path of lengths $\{0 + k \bmod 2, 2 + k \bmod 2, ..., k\}$ then, $\beta(A(\Gamma)^{n+1}) = \beta(A(\Gamma)^{n-1})$.
\end{proof}
\begin{proof}
\subsubsection*{$2 \implies 1$}
Assume $\exists\ k : \beta(A(\Gamma)^{k+1}) = \beta(A(\Gamma)^{k-1})$ where $k$ is the minimum such number to satisfy this property. Additionally, assume that $\mathfrak{D}(\Gamma) = n$ where $k \ne n$. Deny the statement, let $k \ne n$. By $1 \implies 2$ we know that since $\beta(A(\Gamma)^{k+1}) = \beta(A(\Gamma)^{k-1})$ then $k > \mathfrak{D}(\Gamma)$. However, since $k$ is the minimum such number, and $k \ne n$, the assumption is false. Thus, $k = n$.
\end{proof}

Hence we achieve an algebraic condition that can be evaluated to show exactly the length of the longest path. By forming a data structure of ``potential optional lengths of longest path'' we can perform a binary search on them to efficiently search for the longest path. This is detailed further in Section 4 and 6. Theorem 1 allows us to know what the length of the longest path is. However, there is also interest in \textit{generating} the longest path. Hence, Corollary 1.1 shows one such method that can be performed in iterations that will allow us to generate such paths.

\begin{corollary}
    Let $\Gamma$ be a Tree and denote $\gamma_{a,b}$ to be the unique path between vertices $a$ and $b$ in $\Gamma$. Assume that $\beta(A(\Gamma)^{n})_{i,j} - \beta(A(\Gamma)^{n-2})_{i,j} = 1$. Then, 
    $$\beta(A(\Gamma)^{n-1})_{i,k} - \beta(A(\Gamma)^{n-3})_{i,k} = 1 \text{ and } A(\Gamma)_{j,k} = 1 \iff \gamma_{v_i, v_j} = \gamma_{v_i, v_k} \cup (v_j)$$
\end{corollary}

\begin{proof}$\implies$\newline
    Assume $\beta(A(\Gamma)^{n})_{i,j} - \beta(A(\Gamma)^{n-2})_{i,j} = 1$ and that $A(\Gamma)_{j,k} = 1$. Let $\gamma_{a,b}$ denote the unique path between vertices $a$ and $b$. Hence, by Lemma 3, we know that if $\beta(A(\Gamma)^{n})_{i,j} = 1$, then $n \in \{\mathscr{L}(\gamma_{v_i,v_j}) + 2m \ |\ m \in \mathbb{N} \}$. Given that we assume $\beta(A(\Gamma)^{n})_{i,j} - \beta(A(\Gamma)^{n-2})_{i,j} = 1$. Then, it must be the case that $\beta(A(\Gamma)^{n})_{i,j} = 1$ and $\beta(A(\Gamma)^{n-2})_{i,j} = 0$. Therefore, $n-2 \notin \{\mathscr{L}(\gamma_{v_i,v_j}) + 2m \ |\ m \in \mathbb{N} \}$. Hence, we see that $n = \mathscr{L}(\gamma_{v_i,v_j})$. Now, assume that $\beta(A(\Gamma)^{n-1})_{i,k} - \beta(A(\Gamma)^{n-3})_{i,k} = 1 \text{ and } A(\Gamma)_{j,k} = 1$. Then, similarly we know that $n-1 = \mathscr{L}(\gamma_{v_i, v_k})$ and $\{v_k, v_j\} \in \mathbb{E}(\Gamma)$. Since $\Gamma$ is a tree, then there are unique paths for every pair of vertices. Hence, we can construct $\gamma_{v_i, v_j} = \gamma_{v_i, v_k} \cup (v_j)$.\\

\noindent \textit{Proof. }$\Longleftarrow$\newline
    Assume that $\beta(A(\Gamma)^{n})_{i,j} - \beta(A(\Gamma)^{n-2})_{i,j} = 1$ and $\gamma_{v_i, v_j} = \gamma_{v_i, v_k} \cup (v_j)$. Then, by the same reasoning as the first proof, we know that $n = \mathscr{L}(\gamma_{v_i,v_j})$. Since $\gamma_{v_i, v_j} = \gamma_{v_i, v_k} \cup (v_j) = (v_i, ..., v_k, v_j)$ then that means $\{v_k, v_j\} \in \mathbb{E}(\Gamma)$ so we know that $A(\Gamma)_{j,k} = 1$. Since $\mathscr{L}(\gamma_{v_i, v_k}) = \mathscr{L}(\gamma_{v_i, v_j}) - 1$ then $\mathscr{L}(\gamma_{v_i, v_k}) = n-1$. Therefore, $n-3 \notin \{\mathscr{L}(\gamma_{v_i, v_k}) + 2m\ |\ m \in \mathbb{N}\}$ implying that we have $\beta(A(\Gamma)^{n-1})_{i,k} - \beta(A(\Gamma)^{n-3})_{i,k} = 1$.
\end{proof}
By taking $n = \mathfrak{D}(\Gamma)$, we can then start generating the longest path beginning from the end points, and through $n$ iterations then find the overall path. This is the idea behind Algorithm 8 which is detailed later in Section 5 and 6. Similarly to Trees, we can also use a similar algebraic condition for Uniform Block Graphs. Recall that we consider Uniform Block Graphs to be a special case of Block Graphs so that every Block is a clique of the same size. First however, we find the exact length of the longest path of Uniform Block Graphs.

\begin{theorem}
Let $J_m$ be the $m \times m$ matrix of all 1's. Let $\Gamma$ be a Uniform Block Graph, then TFAE:
\begin{enumerate}
    \item The length of the Longest Path of $\Gamma$ is $n\cdot (\omega(\Gamma) - 1) = n \cdot (m - 1)$
    \item $\beta(A(\Gamma)^{n}) = J_{|\mathbb{V}(\Gamma)|} :$ n is the minimum such number
\end{enumerate}
The Length of the Longest Path is found by the following
\begin{align*}
    \mathscr{LP}(\Gamma) &= min \{n \ |\  \beta(A(\Gamma)^{n})  = J_{|\mathbb{V}(\Gamma)|}\} \cdot (\omega(\Gamma) - 1)
\end{align*}
\end{theorem}
\begin{proof}$1 \implies 2$\newline
Assume $\mathscr{LP}(\Gamma) = n\cdot(m-1)$ is the length of the longest path. Define $M := \{v_1, v_2, ..., v_d\}$ so that $\text{deg}(v_i) = \Delta(\Gamma) = 2(m-1)$. Hence, by construction we know $M$ is a vertex separator of $\Gamma$. Thus:
\begin{align*}
    \forall\ \mathbb{V}(B_p)\cap \mathbb{V}(B_q) \ne \phi, \ \exists! v_i \in M & \text{ such that } \{v_i\} = \mathbb{V}(B_p) \cap \mathbb{V}(B_q)\ \text{ for all blocks } B_p, B_q \subseteq \Gamma\\
    &\implies \bigcup_{i=1}^dN(v_i) \cup M = \mathbb{V}(\Gamma)
\end{align*}
We can then see that any (minimal) length path is of the form:
\begin{align*}
    \gamma_{uw}^{min} = (u, ..&., w) \text{ such that } \mathscr{L}(\gamma_{uw}^{min}) = \partial(u,w) \text{ for } u,w \in \mathbb{V}(\Gamma)\\
    &\implies \exists\ v_u,v_w \in M \text{ whereby } u \in N(v_u), w \in N(v_w)\\
    &\implies \gamma_{uw}^{min} = (u, v_u, v_i, ..., v_k , ..., v_j, v_w, w) \text{ such that } v_i, v_k, v_j \in M\\
    & \ \ \ \ \ \ \ \ \ \ \ \ \ \ \ \ \ \ \ \ \ \ \text{ and } v_i \in N(v_u), v_j \in N(v_w)\ \forall\ u, w \in \mathbb{V}(\Gamma)
\end{align*}
Therefore the longest (minimal) path $\gamma_l^{min}$ must have:
\begin{align*}
    \exists\ v_a, v_b \in M \text{ such that } \gamma_l^{min} = (a_l&,v_a, ..., v_b, b_l) \text{ whereby } \mathscr{L}(\gamma) = \partial(a_l,b_l) \text{ and } a_l \in N(v_a), b_l \in N(v_b)\\
    &\implies \partial(a_l, b_l) = \partial(v_a,v_b) + 2
\end{align*}
Which leads to the form of the longest path $\gamma_l^{max}$ which must have:
\begin{align*}
    \exists\ v_a, v_b \in M \text{ so that }& \gamma_l^{max}  = \ (\underbrace{a_l,r_1, r_2, ..., r_{m - 2}}_{m - 1\text{ elements }\in N(v_a) \cap \mathbb{V}(B_a)}, v_a, ..., v_k, ..., v_b, \underbrace{t_{m - 2}, ..., t_2, t_1, b_l}_{m - 1 \text{ elements }\in N(v_b) \cap \mathbb{V}(B_b)})\\
    \text{ such that }&  a_l, b_l \in \gamma_l^{min}\  \forall\ r_i \in N(v_a), t_j \in N(v_b) \text{ which forms the same } \forall\ v_k \in M\\
    \implies  \gamma_l^{max} &= (\underbrace{\left[N(v_a)\cap \mathbb{V}(B_a)\right]}_{m - 1\text{ elements}}, v_a, ..., v_{k-1}, \underbrace{\left[N(v_k)\cap \mathbb{V}(B_k)\right]}_{m - 1\text{ elements}}, v_k, ..., v_b, \underbrace{\left[N(v_b)\cap \mathbb{V}(B_b)\right]}_{m - 1\text{ elements}})\\
    \text{following the size of } & \gamma_l^{min} \implies 1 \le k \le \partial(a_l, b_l) = \partial(v_a,v_b) + 2 \text{ giving us }\mathscr{L}(\gamma_l^{max}) = \partial(a_l,b_l)\cdot(m-1)
\end{align*}
Given that we assumed $\mathscr{LP}(\Gamma) = n\cdot(m-1)$ Then we have 
$$\mathscr{LP}(\Gamma) = n\cdot(m-1) = \mathscr{L}(\gamma_l^{max}) = \partial(a_l,b_l)\cdot(m-1) \implies n = \partial(a_l,b_l)$$
Consequently, since every vertex $u$ belongs to a $K_3$ subgraph of their block $B_u$, then
\begin{align*}
    \forall\ u, w \in \mathbb{V}(\Gamma),& \exists\ B_u, B_w \subseteq \Gamma \text{ maximal 2-connected} \text{ so that } B_u \cong B_w \cong K_m\\
    &\implies \exists\ T_u \subseteq B_u \text{ and } T_w \subseteq B_w \text{ whereby } T_u \cong T_w \cong K_3 \text{ so that } u \in T_u \text{ and } w \in T_w
\end{align*}
Implying that we can always find a walk of length $\partial(a_,b_l) = n$ for any pair of vertices
\begin{align*}
    \omega = (u, u_T, ..., w_T, w) &\text{ so that } \mathfrak{L}(\omega) = \partial(a_l,b_l) = n \text{ and } u_T \in T_u, w_T \in T_w\ \forall\ u, w \in \mathbb{V}(\Gamma)\\
    &\implies \exists\ \omega = (u, ..., w) \text{ where } \mathfrak{L}(\omega) = n\ \forall\ u, w \in \mathbb{V}(\Gamma)\\
    &\implies \beta(A(\Gamma)^n)_{i,j} = 1 \ \forall\ i,j \in \mathbb{Z}_{|\mathbb{V}(\Gamma)|+1}^*\\
    &\implies \beta(A(\Gamma)^n) = J_{|\mathbb{V}(\Gamma)|}
\end{align*}
We conclude that $n$ is the minimum such number since the length of $\gamma_l^{min}$ is $n$, hence if $\exists\ k < n$ so that $\beta(A(\Gamma)^{k}) = J_{|\mathbb{V}(\Gamma)|}$ then there is a walk of length $k$ between every two vertices, which is a contradiction since $\partial(a_l, b_l) = n > k$\\

\noindent \textit{Proof. }$2 \implies 1$\newline
Assume $\beta(A(\Gamma)^n) = J_{|\mathbb{V}(\Gamma)|}$ where $n$ is the minimum such number. Therefore we know that:
$$\exists\ \omega = (u, ..., w) \text{ with } \mathfrak{L}(\omega) = n \ \forall\ u,w \in \mathbb{V}$$
From the proof of $1 \implies 2$, we know that for $\gamma_l^{min}$, we have:
\begin{align*}
    \gamma_l^{min} = (a_l, ..., b_l) \text{ with } \mathscr{L}(\gamma_l^{min}) = \partial(a_l,b_l) = n
\end{align*}
Similarly from $1 \implies 2$ we know for $\gamma_l^{max}$ we have:
\begin{align*}
    \gamma_l^{max} & = \ (\underbrace{a_l,r_1, r_2, ..., r_{m - 2}}_{m - 1\text{ elements }\in N(v_a) \cap \mathbb{V}(B_a)}, v_a, ..., v_k, ..., v_b, \underbrace{t_{m - 2}, ..., t_2, t_1, b_l}_{m - 1 \text{ elements }\in N(v_b) \cap \mathbb{V}(B_b)})\\
    \text{ such that }& a_l, b_l \in \gamma_l^{min}\  \forall\ r_i \in N(v_a), t_j \in N(v_b) \text{ which forms the same } \forall\ v_k \in M\\
    &= (\underbrace{\left[N(v_a)\cap \mathbb{V}(B_a)\right]}_{m - 1\text{ elements}}, v_a, ..., v_{k-1}, \underbrace{\left[N(v_k)\cap \mathbb{V}(B_k)\right]}_{m - 1\text{ elements}}, v_k, ..., v_b, \underbrace{\left[N(v_b)\cap \mathbb{V}(B_b)\right]}_{m - 1\text{ elements}})
\end{align*}
following the size of $\gamma_l^{min} \implies 1 \le k \le \partial(a_l, b_l) = \partial(v_a,v_b) + 2.$ Therefore $\mathscr{L}(\gamma_l^{max}) = \partial(a_l,b_l)\cdot(m-1).$

\noindent
Finally, using the fact that $n = \partial(a_l,b_l)$ we end up with:
\begin{align*}
    n = \partial(a_l,b_l) & \implies n\cdot(m-1) = \partial(a_l, b_l)\cdot(m-1) = \mathscr{L}(\gamma_l^{max})\\ 
    &\implies \mathscr{L}(\gamma_l^{max}) = \mathscr{LP}(\Gamma) = n\cdot(m-1)
\end{align*}
\end{proof}
The underlying structure of Uniform Block Graphs is as though you define it to be a tree where every vertex is a block of the same size. Hence, there is a similarity in the way we can find a similar condition as in Theorem 1. We exploit the fact that all blocks are of the same size, and can then simply multiply it with the length of the longest path of the underlying tree in the Uniform Block Graph. To emphasize this observation, we define a ``chain'' to be a sequence of blocks in a Block Graph.
\newpage 

\begin{definition}
    Define a chain, $\mathcal{C}$ of a Block Graph, $\Gamma$, to be a sequence of blocks
    $$\mathcal{C} = (B_1, B_2, ..., B_{|\mathcal{C}|}) \text{ so that } B_k \subseteq \Gamma \text{ maximal 2-connected}$$
\end{definition}
 Therefore, we can think of the longest path of the underlying tree as the length of the longest chain. In particular, we can also take the minimal spanning tree of the Uniform Block Graph, and see that the longest path of the tree models the longest chain. Hence we can find the length of the longest chain by this result.
\begin{corollary}
Let $J_m$ be the $m \times m$ matrix of all 1's. Let $\Gamma$ be a Uniform Block Graph. Then, the size $\mathcal{L}$ of the longest chain $\mathcal{C}_l = (B_1, B_2, ..., B_{\mathcal{L}})$ in $\Gamma$ can be found by:
\begin{align*}
    \mathcal{L} &= min \{n \ |\  \beta(A(\Gamma)^{n})  = J_{|\mathbb{V}(\Gamma)|}\}
\end{align*}
\end{corollary}

\begin{proof}$ $\newline 
From $1 \implies 2$ of the Theorem 2 on Uniform Block Graphs, we know the longest (minimal) path $\gamma_l^{min}$ is of the form
\begin{align*}
    \gamma_l^{min} &= (a_l,v_a, ..., v_b, b_l) \text{ so that } \mathscr{L}(\gamma) = \partial(a_l,b_l) \text{ and } a_l \in N(v_a), b_l \in N(v_b)\\
    &\ \ \ \ \ \ \ \ \ \ \ \ \ \ \ \ \ \ \ \ \ \ \ \ \implies a_l, v_a \in B_1, b_l, v_b \in B_s\\
    &\ \ \ \ \ \ \ \ \ \ \ \ \ \ \ \ \ \ \ \ \ \ \ \ \implies v_k \in B_{k-1} \text{ and } v_k \in B_{k}\ \forall\ v_k \in \gamma_l^{min}
\end{align*}
Therefore we can see that for blocks $B_2, ..., B_{\partial(a_l, b_l) - 1}$ each contains two $v_{k-1},v_{k}$ whereby:
\begin{align*}
    &2 \le k \le \partial(a_l,b_l) - 1 \implies \exists\ v_{k-1}, v_{k} \in \gamma_l^{min} \text{ such that } v_{k-1}, v_k \in B_k
\end{align*}
Not counting blocks $B_1, B_{\partial(a_l,b_l)}$ since they are the `leaf' blocks of $\Gamma$, they only have a single $v_1 \in B_1, v_{\partial(a_l,b_l)} \in B_{\partial(a_l,b_l)}$. We also know we can count every block using $\gamma_l^{min}$ since $\exists\ v \in \gamma_l^{min}\ \forall\ v \in B_v$ for some $B_v \in \mathcal{C}_l$. Hence, we can see that the number of blocks corresponds directly to:
$$\mathcal{L} = 2 + \sum_{\forall v_k \in \gamma_l^{min}\setminus \{v_1, v_{\partial(a_l,b_l)}\}}1 = \partial(a_l,b_l)$$
From $1\implies 2$ of Theorem 2 on Uniform Block Graphs we know that $n = \partial(a_l, b_l)$ is the minimum number for which $\beta(A(\Gamma)^n) = J_{|\mathbb{V}(\Gamma)|}$. Therefore, we can conclude that
\begin{align*}
    \mathcal{L} = \partial(a_l,b_l) &= min \{n \ |\  \beta(A(\Gamma)^{n})  = J_{|\mathbb{V}(\Gamma)|}\}
\end{align*}
\end{proof}

Since Corollary 2.1 provides the length of the longest chain of blocks, as can be seen from the proof, nothing is restricting us to a single unique clique for all blocks. Therefore, we can use it for Block Graphs too, though the algorithm will have to work differently with more steps as compared to Uniform Block Graphs. In particular, given that Uniform Block Graphs have chains with unique cliques of order $m$, then there is no requirement to compute the order of every approached block in the longest chain. Meanwhile, Block Graphs have non-unique cliques for blocks. 

Therefore we require a method to compute the order of the clique for every clique in the longest chain, which hence adds some intricacy when generalizing to Block Graphs. This is why a distinction is made for results on Uniform Block Graphs versus Block Graphs. Firstly, to find the length of the longest path of Block Graphs, we need a description of how we can identify when we have a path of a certain length just as we have for Trees.
\newpage
\begin{lemma}
    Let $\Gamma$ be a Block Graph, then for $n \ge 1$
    $$\beta(A(\Gamma)^n)_{i,j} = 1 \iff \exists\ \gamma = (v_i,...,v_j) \text{ such that } \mathscr{L}(\gamma) \in \{1, 2, ..., n\}$$
\end{lemma}
\begin{proof}$\implies$\newline
    Assume $\beta(A(\Gamma)^n)_{i,j} = 1$, then there is a walk $\omega = (v_i, ..., v_j)$ of length $n$. Given that every block $B_k$ in $\Gamma$ is a clique, then for every $v \in \mathbb{V}(\Gamma)$, there is some $K_3 \simeq T_v \subseteq \Gamma$ such that $v \in \mathbb{V}(T_v)$. Hence, every vertex is part of some $K_3$ induced subgraph of $\Gamma$. Therefore, for any two vertices $a, b \in \mathbb{V}(\Gamma)$, there will always be a walk of length $k$ for any $k \ge \partial(a, b)$ since we can walk an odd or even number of times by traversing the $K_3$ neighborhood of $a$ or $b$ respectively. Hence, since there is a walk of length $n$ from $v_i$ to $v_j$, then $1 \le \partial(v_i, v_j) \le n$. Therefore, $\exists\ \gamma = (v_i, ..., v_j)$ such that $\mathscr{L}(\gamma) \in \{1, ..., n\}$.\\

\noindent\textit{Proof. }$\Longleftarrow$\newline
    Assume $\exists\ \gamma = (v_i,...,v_j) \text{ such that } \mathscr{L}(\gamma) \in \{1, 2, ..., n\}$. Given that every block $B_k$ in $\Gamma$ is a clique, then for every $v \in \mathbb{V}(\Gamma)$, there is some $K_3 \simeq T_v \subseteq \Gamma$ such that $v \in \mathbb{V}(T_v)$. Hence, every vertex is part of some $K_3$ induced subgraph of $\Gamma$. This implies that for any chosen length $1 \le \mathscr{L}(\gamma) \le n$, there will be a walk $\omega = (v_i,...,v_j)$ of length $n$ since we can traverse in the $K_3$ neighborhood of $v_i$ or $v_j$ respectively for odd or even $n$. Hence, $\beta(A(\Gamma)^n)_{i,j} = 1$.
\end{proof}

We can now use Lemma 4 to inform us exactly where we can find lengths of paths bounded to a certain region. Say you choose to find the end points of all paths of lengths bounded below by $m$ and above by $n$. Then, using Lemma 4, we can see when an entry of the booleanized matrix (taken to a certain power) is equal to 1 or equal to 0. More precisely, the question of the existence of a path of length $k$ so that $m < k \le n$ can be solved by checking if $\beta(A(\Gamma)^n)_{i,j} = 1$ or $0$ and also $\beta(A(\Gamma)^m)_{i,j} = 1$ or $0$. We combine this all into one statement as follows.

\begin{lemma}
    Let $\Gamma$ be a Block Graph, then for $m < n$
    $$\beta(A(\Gamma)^n)_{i,j} - \beta(A(\Gamma)^{m})_{i,j} = 1 \iff \exists\ \gamma = (v_i, ..., v_j) \text{ such that } \mathscr{L}(\gamma) \in \{m + 1, ..., n\}$$
\end{lemma}
\begin{proof}$\implies$\newline
    Assume $\beta(A(\Gamma)^n)_{i,j} - \beta(A(\Gamma)^{m})_{i,j} = 1$, then we know that $\beta(A(\Gamma)^n)_{i,j} = 1$ while $\beta(A(\Gamma)^m)_{i,j} = 0$. Therefore, by Lemma 4, we know that $\exists\ \gamma = (v_i, ..., v_j)$ such that $\mathscr{L}(\gamma) \in \{1, ..., n\}$ however also $\mathscr{L}(\gamma) \notin \{1, ..., m\}$. Therefore, we deduce that $\mathscr{L}(\gamma) \in \{m+1, ..., n\}$.\\

\noindent\textit{Proof. }$\Longleftarrow$\newline
    Assume $\exists\ \gamma = (v_i, ..., v_j) \text{ such that } \mathscr{L}(\gamma) \in \{m + 1, ..., n\}$. Then, we know by Lemma 4 that $\beta(A(\Gamma)^n)_{i,j} = 1$ however $\beta(A(\Gamma)^m)_{i,j} = 0$. Hence, we have $\beta(A(\Gamma)^n)_{i,j} - \beta(A(\Gamma)^m)_{i,j} = 1$.
\end{proof}

Therefore, Lemma 5 allows us to `count' the number of paths of a certain length. By setting $m = n-1$, then we can evaluate $\beta(A(\Gamma)^n)_{i,j} - \beta(A(\Gamma)^{n-1})_{i,j}$. If the value is $1$, then there must be a path of length $n$, otherwise there is none. Since every vertex belongs to a block which is a clique, then for a length of a path $n$, we can evaluate the number of vertices of a certain block. By restricting this to only blocks in the longest chain, then we can find the longest path.

\begin{theorem}
    Let $\Gamma$ be a Block Graph such that $\mathcal{C} = (B_1, B_2, ...)$ is the longest chain of length $\mathcal{L}$. Let $v_\alpha, v_\beta \in \mathbb{V}(\Gamma)$ so that $\beta(A(\Gamma)^\mathcal{L})_{\alpha, \beta} - \beta(A(\Gamma)^{\mathcal{L}-1})_{\alpha, \beta} = 1$, and $\mathscr{LP}(\Gamma)$ be the length of the Longest Path of a given Block Graph, then
    \begin{align*}
    \mathscr{LP}(\Gamma) = \sum_{n=1}^{\mathcal{L}}\sum_{v_i \in B_n}\beta(A(\Gamma)^n)_{\alpha,i} - \beta(A(\Gamma)^{n-1})_{\alpha,i} \text{ such that } \mathcal{L} = min\{k \ | \ \beta(A(\Gamma)^k) = J_{|\mathbb{V}(\Gamma)|}\}
    \end{align*}
\end{theorem}
\begin{proof}$ $\newline
    From the proof of Theorem 2 and Corollary 2.1, we can see with a change to $m$ accounting for dependency on every $B_k \in \mathcal{C}$ we have
    \begin{align*}
        \gamma_l^{max} & = \ (\underbrace{a_l,r_1, r_2, ..., r_{|\mathbb{V}(B_a)| - 2}}_{|\mathbb{V}(B_a)| - 1\text{ elements }\in N(v_a) \cap \mathbb{V}(B_a)}, v_a, ..., v_k, ..., v_b, \underbrace{t_{|\mathbb{V}(B_b)| - 2}, ..., t_2, t_1, b_l}_{|\mathbb{V}(B_b)| - 1 \text{ elements }\in N(v_b) \cap \mathbb{V}(B_b)})\\
    \text{such that }&  a_l, b_l \in \gamma_l^{min}\  \forall\ r_i \in N(v_a), t_j \in N(v_b) \text{ which forms the same } \forall\ v_k \in M\\
    \implies & \gamma_l^{max} = (\underbrace{\left[N(v_a)\cap \mathbb{V}(B_a)\right]}_{|\mathbb{V}(B_a)| - 1\text{ elements}}, v_a, ..., v_{k-1}, \underbrace{\left[N(v_k)\cap \mathbb{V}(B_k)\right]}_{|\mathbb{V}(B_k)| - 1\text{ elements}}, v_k, ..., v_b, \underbrace{\left[N(v_b)\cap \mathbb{V}(B_b)\right]}_{|\mathbb{V}(B_b)| - 1\text{ elements}})
    \end{align*}
    following the size of $\gamma_l^{min} \implies 1 \le k \le \partial(a_l, b_l) = \partial(v_a,v_b) + 2$. Therefore $\mathscr{L}(\gamma_l^{max}) = |\mathbb{V}(B_1)| - 1 + |\mathbb{V}(B_2)| - 1 + ... + |\mathbb{V}(B_{\partial(a_l,b_l)})| - 1$. Given that $v_\alpha$ belongs to a block in the end-points of the longest chain, $\mathcal{C}$, then by Lemma 4 and 5 we know that
    $$|\mathbb{V}(B_k)| = 1 + \sum_{v_i \in B_k}\beta(A(\Gamma)^{\partial(v_\alpha, v_i)})_{\alpha,i} - \beta(A(\Gamma)^{\partial(v_\alpha, v_i)-1})_{\alpha,i}$$
    Since $v_\alpha$ has a path of length $\partial(v_\alpha, v_i)$ for every $v_i \in B_k$ for all $B_k \in \mathcal{C}$. Therefore, by summing over every $B_k \in \mathcal{C}$ we can find the longest path as
    $$\mathscr{L}(\gamma_l^{max}) = \mathscr{LP}(\Gamma) = \sum_{n=1}^{\mathcal{L}}\sum_{v_i \in B_n}\beta(A(\Gamma)^n)_{\alpha,i} - \beta(A(\Gamma)^{n-1})_{\alpha,i}$$
\end{proof}
We are able to exactly compute the length of the longest path of Block Graphs by this method. As discussed, the method that we compute for Block Graphs is more complex since it requires an input of the longest chain. It is possible to find the longest chain by taking the spanning tree of the Block Graph, then generating the longest path of the spanning tree. Hence, the longest path would have vertices uniquely corresponding to a pair of blocks in the longest chain. Allowing us to generate the longest chain then use that as input to find the longest path of the original Block Graph. To summarize this so we can use it in our algorithms, we can see that
\begin{remark}
    For a given Block Graph, $\Gamma$ with a spanning tree $T \subseteq \Gamma$, the size of the longest chain in $\Gamma$ is equal to the length of the longest path of $T$.
\end{remark}

We can use Kruskal's algorithm  to generate the minimal spanning tree, which is of complexity $O(m \cdot log(m))$ where $m$ is the size of the graph \cite{Kruskal, Kleinberg}. This allows us to efficiently find the longest path of Block Graphs algebraically. This will be detailed in a later section, specifically Section 4.

For our final class of graphs, we can algebraically compute the length of the longest path of Directed Acyclic Graphs using a similar condition to Trees, given their acyclic nature. We exploit the fact that the adjacency matrices of Directed Acyclic Graphs are nilpotent, and show how using the `booleanized' matrix power is more efficient than computing by classic matrix multiplication. Firstly, we need the ability to bound the length of paths for a given booleanized matrix power.
\begin{lemma}
    Let $\mathcal{D}$ be a Directed Acyclic Graph, then for $n \ge 1$
    $$\beta(A(\mathcal{D})^{n})_{i,j} = 1 \iff \exists\ \gamma = (v_i, ..., v_j) \text{ such that } \mathscr{L}(\gamma) = n$$
\end{lemma}
\begin{proof}$\implies$\newline
    Assume that $\beta(A(\mathcal{D})^{n})_{i,j} = 1$. Then, there is a walk of length $n$ from $v_i$ to $v_j$. Therefore, there exists a path $\gamma = (v_i, ..., v_j)$. Given that the graph is directed and acyclic, then there exists no other walk from $v_i$ to $v_j$ since it is not possible to travel backwards. Hence, there is a unique walk and therefore a unique path of length $n$

\noindent\textit{Proof. }$\Longleftarrow$\newline
    Assume $\exists\ \gamma = (v_i, ..., v_j) : \mathscr{L}(\gamma) = n$, then there is a walk of length $n$. Hence, $\beta(A(\Gamma)^n)_{i,j} = 1$.
\end{proof}

Directed Acyclic Graphs give us the `nicest' structure for the booleanized matrices, as walks and paths are uniquely corresponded, then the matrix powers give us all the information we need about the length of paths. Hence, you can choose some power $n$ where there are no paths of that length, then traverse backwards to then find the point at which you have the maximum length as follows.

\begin{theorem}
    Let $\mathcal{D}$ be a Directed Acyclic Graph, then TFAE:
    \begin{enumerate}
    \item The length of the Longest Path of $\mathcal{D}$ is $n$
    \item $\beta(A(\mathcal{D})^{n}) = \left[\mathbf{0}\right]_{|\mathbb{V}(\mathcal{D})|} :$ n is the minimum such number
\end{enumerate}
    Equivalently, to compute the value, the Length of the Longest Path is found by the following
    $$\mathfrak{D}_{_{dag}}(\mathcal{D}) = min \{n\ |\ \beta(A(\mathcal{D})^n) = \left[\mathbf{0}\right]_{|\mathbb{V}(\mathcal{D})|}\}$$
\end{theorem}
\begin{proof}
    \subsubsection*{$1 \implies 2$}
    Assume $\mathfrak{D}_{_{dag}}(\mathcal{D}) = n$ Therefore $ \forall\ u, v \in \mathbb{V}(\mathcal{D})\ \nexists\ \gamma = (u, ..., v)$ such that $ \mathscr{L}(\gamma) = n + 1$. Hence by Lemma 2 and 3, we can see that between any pair of vertices following the direction of the $DAG$, we have $\nexists\ \omega = (u, ..., v) $ so that $ \mathfrak{L}(\omega) = n + 1 $ since $ \mathfrak{D}_{_{dag}}(\mathcal{D}) = n$. Therefore, $n$ would be the index of $A(\mathcal{D}) \implies \beta(A(\mathcal{D})^{n}) = \left[\mathbf{0}\right]_{|\mathbb{V}(\mathcal{D})|}$. We know $n$ is the minimum such number since $\exists \ \omega = (u, ..., v) : \mathfrak{L}(\omega) = \mathfrak{D}_{_{dag}}(\mathcal{D}) = n$ for some $u, v \in \mathbb{V}(\mathcal{D})$.
\end{proof}
\begin{proof}
    \subsubsection*{$2 \implies 1$}
    Assume $\exists\ k : \beta(A(\mathcal{D})^k) = \left[\mathbf{0}\right]_{|\mathbb{V}(\mathcal{D})|}$ where $k$ is the minimum such number to satisfy this property. Additionally, assume that $\mathfrak{D}_{_{dag}}(\mathcal{D}) = n$ where $k \ne n$. By $1 \implies 2$ we know that since $\beta(A(\mathcal{D})^k) = \left[\mathbf{0}\right]_{|\mathbb{V}(\mathcal{D})|}$ then $k > \mathfrak{D}_{_{dag}}(\mathcal{D})$. However, since $k$ is the minimum such number, and $k \ne n$, the assumption is false. Thus, $k = n$ must be true.
\end{proof}

Since the matrix is nilpotent, then we can easily see that $\forall\ k \ge \mathfrak{D}_{dag}(\mathcal{D})$, then $\beta(A(\mathcal{D})^k) = [\textbf{0}]_{|\mathbb{V}(\mathcal{D})|}$. In fact, given that it is all 0-entries, then we technically do not require booleanizing the matrix. However, as we have mentioned previously, computing the product based on Algorithm 1 is more efficient than classic matrix multiplication algorithms. Hence, we decide to use the booleanized form to emphasize that fact. Just as we have done with Trees, we can also generate the longest paths of a Directed Acyclic Graph using the following result.

\begin{corollary}
    Let $\mathcal{D}$ be a Directed Acyclic Graph and denote $\gamma_{a,b}$ to be the unique path between vertices $a$ and $b$ in $\mathcal{D}$. Assume that $\beta(A(\mathcal{D})^n)_{i,j} = 1$. Then,
    $$\beta(A(\mathcal{D})^{n-1})_{i,k} = 1 \text{ and } A(\mathcal{D})_{j,k} = 1 \iff \gamma_{v_i, v_j} = \gamma_{v_i, v_k} \cup (v_j)$$
\end{corollary}
\begin{proof}$\implies$\newline
    Assume that $\beta(A(\mathcal{D})^{n-1})_{i,k} = 1 \text{ and } A(\mathcal{D})_{j,k} = 1$. By Lemma 6 we know that $\mathscr{L}(\gamma_{v_i,v_j}) = n$ and $\mathscr{L}(\gamma_{v_i,v_k}) = n-1$ in addition to $(v_k,v_j) \in \mathbb{E}(\mathcal{D})$. Therefore, since there are only unique paths between vertices then the only way to traverse to $v_j$ is by $(v_k, v_j)$. Hence, $\gamma_{v_i, v_j} = \gamma_{v_i, v_k} \cup (v_j)$.\\

\noindent\textit{Proof. }$\Longleftarrow$\newline
    Assume that $\gamma_{v_i, v_j} = \gamma_{v_i, v_k} \cup (v_j)$. Since $\beta(A(\mathcal{D})^n)_{i,j} = 1$ then $\mathscr{L}(\gamma_{v_i,v_j}) = n$, since we have only unique paths and $(v_k, v_j) \in \mathbb{E}(\mathcal{D})$ due to being adjacent in $\gamma_{v_i,v_j}$, then we know that $\mathscr{L}(\gamma_{v_i,v_k}) = n-1$. Hence, $\beta(A(\mathcal{D})^{n-1})_{i,k} = 1$.
\end{proof}

\newpage
\section{Longest Path Length Algorithms}
The following two algorithms are a binary search for the length of the longest path of Trees and Directed Acyclic Graphs proven correct as a result of Theorem 1 and Theorem 4.

\begin{algorithm}[H]
\caption{Binary Diameter Search $\left[Tree\right]$}\label{euclid}
\begin{algorithmic}[1]
\Procedure{$\mathfrak{D}$}{$\Gamma, l = 0, h = |\mathbb{E}(\Gamma)|, \epsilon = 1$}
\State $m \gets \lfloor \frac{l+h}{2}\rfloor$
    \If{$h - l \le \epsilon$}
        \State $\text{return}\ (l, h)$
    \EndIf
    \If{$l = h$}
        \State $\text{return } m$
    \EndIf
    \If{$\beta(A(\Gamma)^{m+1}) = \beta(A(\Gamma)^{m-1})$}
        \State $\text{return } \mathfrak{D}(\Gamma, l,m)$
    \Else
        \State $\text{return } \mathfrak{D}(\Gamma, m,h)$
    \EndIf
\EndProcedure
\end{algorithmic}
\end{algorithm}

\vspace{-20pt}

\begin{algorithm}[H]
\caption{Binary Diameter Search $\left[Directed\ Acyclic\ Graph\right]$}\label{euclid}
\begin{algorithmic}[1]
\Procedure{$\mathfrak{D}_{_{dag}}$}{$\mathcal{D}, l = 0, h = |\mathbb{E}(\mathcal{D})|, \epsilon = 1$}
\State $m \gets \lfloor \frac{l+h}{2}\rfloor$
    \If{$h - l \le \epsilon$}
        \State $\text{return}\ (l, h)$
    \EndIf
    \If{$l = h$}
        \State $\text{return } m$
    \EndIf
    \If{$\beta(A(\mathcal{D})^{m}) = \left[\mathbf{0}\right]_{|\mathbb{V}(\mathcal{D})|}$}
        \State $\text{return } \mathfrak{D}_{_{dag}}(\mathcal{D}, l,m)$
    \Else
        \State $\text{return } \mathfrak{D}_{_{dag}}(\mathcal{D}, m,h)$
    \EndIf
\EndProcedure
\end{algorithmic}
\end{algorithm}

For both Uniform Block Graphs and Block Graphs, we use a binary search method to find the length of the longest chain of blocks proven correct as a result of Corollary 2.1.

\begin{algorithm}[H]
\caption{Binary Search for the Length of the Longest Chain $\left[Block\ Graph\right]$}\label{euclid}
\begin{algorithmic}[1]
\Procedure{$Longest\_Chain\_Length$}{$\Gamma, l = 0, h = |\mathbb{E}(\Gamma)|, \epsilon = 1$}
\State $m \gets \lfloor \frac{l+h}{2}\rfloor$
    \If{$h - l \le \epsilon$}
        \State $\text{return}\ (l, h)$
    \EndIf
    \If{$l = h$}
        \State $\text{return } m$
    \EndIf
    \If{$\beta(A(\Gamma)^{m}) = J_{|\mathbb{V}(\Gamma)|}$}
        \State $\text{return } Longest\_Chain\_Length(\Gamma, l,m)$
    \Else
        \State $\text{return } Longest\_Chain\_Length(\Gamma, m,h)$
    \EndIf
\EndProcedure
\end{algorithmic}
\end{algorithm}
Since Uniform Block Graphs are a special case, we can simply find the length of the longest chain, and take its product with the longest path of the block size subtracting one. It is proven correct as a result of Theorem 2.

\begin{algorithm}[H]
\caption{Binary LP Search $\left[Uniform\ Block\ Graph\right]$}\label{euclid}
\begin{algorithmic}[1]
\Procedure{$\mathscr{LP}_U$}{$\Gamma, l = 0, h = |\mathbb{E}(\Gamma)|, \epsilon = 1$}
\State $\text{return } Longest\_Chain\_Length(\Gamma, l, h, \epsilon) \cdot (\omega(\Gamma) - 1)$
\EndProcedure
\end{algorithmic}
\end{algorithm}

As previously mentioned, to find the length of the longest path of general Block Graphs, we require the longest chain as input. To find the longest chain, we can use Kruskal's minimal spanning tree algorithm. We denote this as a procedure $Kruskal\_MST(\Gamma)$ taking in a block graph and returning the minimal spanning tree. We know that $Kruskal\_MST \in O(m log(m))$ where $m$ is the size of $\Gamma$ \cite{Kruskal, Kleinberg}. In addition, to generate the longest path of the minimal spanning tree, we may use our algorithm $Generate\_\mathfrak{D}(\Gamma)$ in Algorithm 8 of Section 5 or Dijkstra's longest path algorithm in $O(n)$ where $n$ is the order of the spanning tree \cite{Bulterman}. Denote the procedure as $Longest\_Path\_of\_Tree(T)$.

Furthermore, we know that the longest path of the minimal spanning tree of a Block Graph must have vertices of the vertex separator of $\Gamma$ in addition to two nodes at the end-points of the longest chain. Hence, given the longest path of the minimal spanning tree, $\gamma_l^{min} = (v_1, ..., v_k)$, then we know that $v_1$ and $v_k$ must be in `leaf-blocks' and every $v_i$ so that $1 \neq i \neq k$ is an element of the vertex separator of the Block Graph. Since every element of the vertex separator is part of more than one block, what we can do is check every edge $\{v_i, v_{i+1}\}$ belonging to $\gamma_l^{min}$ which corresponds uniquely to only one block. Through that, we can then search for which block contains that edge, and that's the block that must be in the longest chain. Otherwise, it is also possible to take the intersection of the neighborhood of each pair of vertices as $N(v_i) \cap N(v_{i+1}) \cup \{v_i, v_{i+1}\}$ which will be the vertex set of the corresponding block.

\begin{algorithm}[H]
\caption{Generate the Longest Chain of a Block Graph $\left[Block\ Graph\right]$}\label{euclid}
\begin{algorithmic}[1]
\Procedure{$Generate\_\mathcal{C}_l$}{$\Gamma$}
\State $T \gets Kruskal\_MST(\Gamma)$
\State $\gamma_l^{min} \gets Longest\_Path\_of\_Tree(T)$
\For{$\{v_i,v_{i+1}\} \in \gamma_l^{min}$}
    \State $V_i \gets N(v_i) \cap N(v_{i+1}) \cup \{v_i, v_{i+1}\}$
    \State $B_i \gets (V_i, \mathbb{E}(V_i))$
    \State $\mathcal{C}_l \gets \mathcal{C}_l \cup (B_i)$
\EndFor
\State $\text{return } \mathcal{C}_l$
\EndProcedure
\end{algorithmic}
\end{algorithm}

Using this procedure, we can then take it as input for the algorithm to actually compute the longest path of a Block Graph.

\begin{algorithm}[H]
\caption{Sum for the Length of the Longest Path $\left[Block\ Graph\right]$}\label{euclid}
\begin{algorithmic}[1]
\Procedure{$\mathscr{LP}$}{$\Gamma$}
    \State $\mathcal{C}_l \gets Generate\_\mathcal{C}_l(\Gamma)$
    \State $sum \gets 0$
    \For{$B \in \mathcal{C}_l$}
        \State $sum \gets sum + |\mathbb{V}(B)| - 1$
    \EndFor
\EndProcedure
\end{algorithmic}
\end{algorithm}
Or we can equivalently use the sum in Theorem 3 as a method to `count up' the value of $|\mathbb{V}(B)| - 1$ for every $B \in \mathcal{C}_l$. The real interest in this case however, is to actually \textit{generate} the longest path. Not only that, but also to generate \textit{every} longest path. In the case of Trees and Directed Acyclic Graphs, we can do this with the result of Corollary 1.1 and Corollary 4.1. For Trees, referring to Corollary 1.1, we can choose $n$ to be the longest path, then generate the longest path `backwards' going down from $n$ to $1$. Similarly, we can do the same for Directed Acyclic Graphs in Corollary 4.1 by choosing $n$ to be the longest path.

\newpage
\section{Path Generating Algorithms}
Once the length of the longest path is achieved, we are able to generate the paths that lead to the longest path. Furthermore, these algorithms give a method to not only generate every longest path, but also paths of any length.
\begin{algorithm}[H]
\caption{Generate $\gamma_{\mathfrak{D}}$ $\left[Tree\right]$}\label{euclid}
\begin{algorithmic}[1]
\Procedure{$Generate\_\mathfrak{D}$}{$\Gamma$}
\State $\gamma \gets \left[\gamma_1, \gamma_2, ...\right]; t \gets 0$
    \For{$i \gets 1, ..., |\mathbb{V}(\Gamma)|$}
        \For{$j \gets i + 1, ..., |\mathbb{V}(\Gamma)|$}
            \If{$(\beta(A(\Gamma)^{\mathfrak{D}(\Gamma)}) - \beta(A(\Gamma))^{\mathfrak{D}(\Gamma)-2})_{i,j} = 1$}
                \State $s \gets i; d \gets j; n \gets \mathfrak{D}(\Gamma) - 1; t \gets t + 1$
                \State $\gamma_t \gets (v_j)$
        \While{$n > 1 \land (\beta(A(\Gamma)^n) - \beta(A(\Gamma))^{n-2})_{s,k} = 1 \land \{v_k, v_d\} \in \mathbb{E}(\Gamma)$}
                \State $\gamma_{t} \gets (v_k) \cup \gamma_{t}$
                \State $k \gets d; n \gets n - 1$
        \EndWhile
            \EndIf
        \EndFor
    \EndFor
    \State $\textbf{return } \gamma$
\EndProcedure
\end{algorithmic}
\end{algorithm}

As we can see from Algorithm 8, we emphasize the fact that this method not only generates a single path of the largest length, but \textit{every path} of that length. This is distinct from Dijkstra's longest path algorithm, as with uniform weights on all edges, the algorithm will only generate a single longest path, while there still might be many that are untouched. To generate all longest paths using Dijkstra's method requires repeating the algorithm multiple times while keeping track of previously-visited nodes and recursively going backwards until all paths are found. This is similar with Directed Acyclic Graphs, as it is possible to use the same algorithm, and we observe the same problem.

\begin{algorithm}[H]
\caption{Generate $\gamma_{\mathfrak{D}_{_{dag}}} \left[Directed\ Acyclic\ Graph\right]$}\label{euclid}
\begin{algorithmic}[1]
\Procedure{$Generate\_\mathfrak{D}_{_{dag}}$}{$\mathcal{D}$}
    \State $\gamma \gets \left[\gamma_1, \gamma_2, ...\right]; t \gets 0$
    \For{$i \gets 1, ..., |\mathbb{V}(\mathcal{D})|$}
        \For{$j \gets 1, ..., |\mathbb{V}(\mathcal{D})|$}
            \If{$(\beta(A(\mathcal{D})^{\mathfrak{D}_{_{dag}}(\mathcal{D})}) - \beta(A(\mathcal{D}))^{\mathfrak{D}_{_{dag}}(\mathcal{D})-1})_{i,j} = 1$}
                \State $s \gets i; d \gets j; n \gets \mathfrak{D}_{_{dag}}(\mathcal{D}) - 1; t \gets t + 1$
                \State $\gamma_t \gets (v_j)$
        \While{$n > 0 \land (\beta(A(\mathcal{D})^n) - \beta(A(\mathcal{D})^{n-1}))_{s,k} = 1 \land \{v_k, v_d\} \in \mathbb{E}(\mathcal{D})$}
                \State $\gamma_t \gets (v_k) \cup \gamma_t$
                \State $k \gets d; n \gets n - 1$
        \EndWhile
            \EndIf
        \EndFor
    \EndFor
        \State $\textbf{return } \gamma$
\EndProcedure
\end{algorithmic}
\end{algorithm}

Both Algorithm 8 and 9 are proven correct by Corollary 1.1 and Corollary 4.1 respectively. This displays the usefulness of algebraic techniques, as the algorithms are a result of the theory. Which can be analyzed and proven more rigorously than compared to purely algorithmic techniques as we have seen with the proof of correctness of Dijkstra's longest path algorithm \cite{Bulterman}. The case for Block Graphs is different, as we do not have unique paths that we can exploit as we have in Corollary 1.1 and Corollary 4.1. However, what we do have is the longest chain, using the longest chain we can then traverse through only the vertices belonging to the longest path of the Block Graph.
\newpage

To be able to generate \textit{all} longest paths in Block Graphs, similarly to Trees, we require \textit{all} longest chains. As mentioned, Dijkstra's longest path algorithm will not generate every longest path of the spanning tree, which is required to generate every longest chain. Hence, we can augment Algorithm 6 to use Algorithm 8 in $Longest\_Path\_of\_Tree$ to generate all longest paths of the spanning tree, that way we can then generate all longest chains, and hence all longest paths of the Block Graph. Here we can simply use Algorithm 6 to generate all longest chains.

\begin{algorithm}[H]
\caption{Generate All Longest Chains of a Block Graph $\left[Block\ Graph\right]$}\label{euclid}
\begin{algorithmic}[1]
\Procedure{$Generate\_all\_\mathcal{C}_l$}{$\Gamma$}
    \State $T \gets Kruskal\_MST(\Gamma)$
    \State $\gamma \gets Generate\_\mathfrak{D}(T)$
    \State $\mathcal{C} \gets []$
    \For{$\gamma_i \in \gamma$}
        \State $\mathcal{C} \gets \mathcal{C} \cup (Generate\_\mathcal{C}_l(\gamma_i))$
    \EndFor
    \State $\text{return } \mathcal{C}$
\EndProcedure
\end{algorithmic}
\end{algorithm}

Therefore we get all longest chains in a Block Graph, using this, we can then generate the longest path of the Block Graph. Given a longest chain $\mathcal{C}_l = (B_1, B_2, ..., B_\mathcal{L})$. We know that we can find the longest path by traversing through all vertices of each Block as long as the vertices belonging to the vertex separator is visited last. In essence, as shown in the proof of Theorem 3, the structure of $\gamma_l^{max}$ is found to be
\begin{align*}
        \gamma_l^{max} = \ (\underbrace{a_l,r_1, r_2, ..., r_{|\mathbb{V}(B_a)| - 2}}_{|\mathbb{V}(B_a)| - 1\text{ elements }\in N(v_a) \cap \mathbb{V}(B_a)}, v_a, ..., v_k, ..., v_b, \underbrace{t_{|\mathbb{V}(B_b)| - 2}, ..., t_2, t_1, b_l}_{|\mathbb{V}(B_b)| - 1 \text{ elements }\in N(v_b) \cap \mathbb{V}(B_b)})
\end{align*}
Where every $v_i$ is in the vertex separator. Meaning that we can begin with $a_l \in B_1$ as the starting vertex, and begin with vertices at the end-block of the longest chain, specifically $B_\mathcal{L}$. Then append all vertices of $B_\mathcal{L}$ making sure to append the vertex element of the vertex separator last, then repeat this process for all blocks in the longest chain. Hence, we generate all longest paths with every possible way to append non-vertex separator vertices of each block. Essentially, by generating one longest path of one chain, then permuting non-vertex separator elements, we get all possible longest paths of a singe longest chain. We repeat this process for all longest chains, and we generate every path of the longest length.

We use a procedure to `permute non-vertex separator elements' to generate every possible longest path for Block Graphs which we denote as $Permute\_non\_vertex\_separator\_elements$. It can be understood that this procedure simply generates all lists so that every neighborhood of every block \textit{other than the vertex separator element} is permuted, and for every permutation it generates a new path. Hence, the possible permutations can be computed. It is precisely the number of non-vertex separator elements in each block. Hence, the leaf-blocks will have $|\mathbb{V}(B)| - 1$ elements to permute, and non-leaf blocks would have $|\mathbb{V}(B)| - 1$ elements to permute. Therefore, assuming we have $\mathcal{L}$ blocks in a single longest chain, we find that we have 
\vspace{-10pt}
$$(|\mathbb{V}(B_1)| - 1)! + (|\mathbb{V}(B_\mathcal{L})| - 1)! + \sum_{i = 2}^{\mathcal{L} - 1}(|\mathbb{V}(B_i)| - 2)!$$
possible longest paths in a one chain. Given that all longest chains must be of the same length, if there are $x$ number of longest chains, and we denote $B_{i,j}$ to be the $i^{th}$ block of the $j^{th}$ chain, then
\vspace{-10pt}
$$\sum_{j = 1}^{x} \left( (|\mathbb{V}(B_{1,j})| - 1)! + (|\mathbb{V}(B_{\mathcal{L}, j})| - 1)! + \sum_{i = 2}^{\mathcal{L} - 1}(|\mathbb{V}(B_{i,j})| - 2)!\right)$$
number of longest paths that can exist and be generated using our algorithms.

It must be noted that in the case of Uniform Block Graphs, we know that $|\mathbb{V}(B_i)| = c$ for some $c \in \mathbb{N}$. Hence, we can compute the number of possible longest paths in a Uniform Block Graph as
\begin{align*}
    \sum_{j = 1}^{x} \left( (c - 1)! + (c - 1)! + \sum_{i = 2}^{\mathcal{L} - 1}(c - 2)!\right) = x \cdot \left( 2(c - 1)! + (\mathcal{L}-2)\cdot(c - 2)!\right)
\end{align*}

Below is the algorithm to generate every possible longest path of a Block Graph. Given that Uniform Block Graphs are a special case, this also provides the solution for them.
\begin{algorithm}[H]
\caption{Generate All Longest Paths of a Block Graph $\gamma_{l}^{max}$ $\left[Block\ Graph\right]$}\label{euclid}
\begin{algorithmic}[1]
\Procedure{$Generate\_all\_\mathscr{LP}$}{$\Gamma$}
    \State $\mathcal{C} \gets Generate\_all\_\mathcal{C}_l(\Gamma);\gamma \gets []$
    \For{$\mathcal{C}_k \in \mathcal{C}$}
        \State $min\_deg\_in\_B_1 \gets \infty$
        \For{$v_l \in B_1(\mathcal{C}_k)$}
            \If{$min\_deg\_in\_B_1 > deg(v_l)$}
                \State $min\_deg\_in\_B_1 \gets deg(v_l)$
                \State $non\_vertex\_separator\_starting\_v \gets v_l$
            \EndIf
        \EndFor
        \State $current\_longest\_path \gets (non\_vertex\_separator\_starting\_v)$
            \For{$n \gets 1, 2, ..., \mathcal{L}$}
                \State $max\_v\_deg \gets 0$
                \For{$v_i \in B_n(\mathcal{C}_K)$}
                    \If{$max\_v\_deg < deg(v_i)$}
                        \State $max\_v\_deg \gets deg(v_i)$
                        \State $v\_separator \gets v_i$
                    \EndIf
                \EndFor
                \For{$v_i \in B_n(\mathcal{C}_K)$}
                    \If{$v_i \neq v\_separator$}
                        \State $current\_longest\_path \gets current\_longest\_path \cup (v_i)$
                    \EndIf
                \EndFor
                \State $current\_longest\_path \gets current\_longest\_path \cup (v\_separator)$
            \EndFor
        \State $\gamma \gets \gamma \cup [current\_longest\_path]; current\_longest\_path \gets ()$
    \EndFor
    \State $\text{return } Permute\_non\_vertex\_separator\_elements(\gamma)$
\EndProcedure
\end{algorithmic}
\end{algorithm}
With this, we can now compute the asymptotic complexity of our algorithms
\section{Complexity Analysis}
Define the order of concern to be, 
$n := |\mathbb{V}(\Gamma)|, m := |\mathbb{E}(\Gamma)|$, 
and $\delta$ to be a small value that depends on the number of bytes of an integer. Additionally, note that for Trees and Directed Acyclic Graphs, $\mathfrak{D}_{_{dag}}(\Gamma) = \mathfrak{D}(\Gamma) = \mathscr{L}(\gamma_l^{max}) = \mathscr{LP}(\Gamma)$. Therefore the results shown below apply to algorithms stated below. Firstly, we find the asymptotic complexity of just computing $\beta(A(\Gamma)^2)$.
\begin{corollary}
    The asymptotic complexity of Algorithm 1 (Booleanization)
    $$\beta(A(\Gamma_1) \times A(\Gamma_2)) \in O(\delta n^{2})$$
\end{corollary}
\begin{proof}$ $\newline
We can simply compute the number of operations as
    $$\beta(A(\Gamma_1) \times A(\Gamma_2)) \in O\left(\sum_{i=1}^{|\mathbb{V}(\Gamma)|}\sum_{j=1}^{|\mathbb{V}(\Gamma)|}\sum_{k=1}^{\delta}1 + \sum_{i=1}^{|\mathbb{V}(\Gamma)|}\sum_{j=1}^{|\mathbb{V}(\Gamma)|}1\right) = O(\delta n^{2} + n^2) = O(\delta n^{2})$$
\end{proof}
Knowing this, we still require computing the booleanized matrix to any power, so we note that
\begin{remark}
    The asymptotic complexity of computing $\beta(A(\Gamma)^k)$ is in $O(\delta n^2 \cdot log_2(k))$ by the use of the binary exponentiation such that
    $$\beta(A(\Gamma)^n) =  
\begin{cases} \beta((A(\Gamma)^2)^{\frac{n}{2}}),& n \text{ is even}\\
    \beta(A(\Gamma) \times (A(\Gamma)^{\frac{n-1}{2}})^2),& n \text{ is odd}\\ \end{cases}$$
\end{remark}
The difficulty here is to compute the asymptotic complexities of the remaining algorithms understanding that the exponent varies for every iteration. Take the case of Algorithm 2 to find the length of the longest path of a Tree. We find that we begin with the exponent $e_1 = \frac{|\mathbb{E}(\Gamma)|}{2}$, however with further iterations, the next exponent depends on where the longest path actually lies, if it is less than half the size, or more than half the size. In the worst case scenario, we can simply assume the maximum number of iterations where the longest path is equal to the size of the graph. In the case of Trees and Directed Acyclic Graphs, that is possible, however it can clearly be seen that in the case of Uniform Block Graphs and Block Graphs, that is not the case and is to be taken as an upper bound. Knowing this, we can compute the asymptotic complexity as follows

\begin{corollary}
    The asymptotic complexity of Algorithm 2, 3, 4, and 5 to find the length of the longest path of Trees, Directed Acyclic Graphs, the longest chain of a Block Graph, and the longest path of Uniform Block Graphs respectively is of the order
    $$\mathfrak{D}, \mathfrak{D}_{_{dag}}, Longest\_Chain\_Length, \mathscr{LP}_U \in O(\delta n^{2} \cdot log_2^2(m))$$
\end{corollary}
\begin{proof}$ $\newline
The search starts at $\frac{|\mathbb{E}(\Gamma)|}{2}$ continuing as a standard Binary Search algorithm with an allocated condition. Thus, a general Binary Search selection results in $k = \left\lceil log_2\left(\frac{|\mathbb{E}(\Gamma)|}{\epsilon}\right)\right\rceil$ where $\epsilon = 1$ for exact results. In every iteration, we perform a booleanized exponentiation of the adjacency matrix, hence every iteration has $\delta n^2\cdot log_2(m)$ operations assuming the worst case scenario. Giving us the final complexity as
\begin{align*}
    O\left(\delta n^{2}\cdot k\cdot log_2(|\mathbb{E}(\Gamma)|)\right) = O\left(\delta n^{2} \cdot log_2(m) \cdot log_2(m)\right) = O\left(\delta n^{2}\cdot log^2_2(m)\right)
\end{align*}
\end{proof}
Note that running Algorithm 2, 3, 4, and 5 by default achieves the exact value of the length of the longest path if $\epsilon = 1$. However, if an error bound is desired without the requirement of an exact value, then error reduces by a factor of $2$ for every iteration $k$.

\begin{remark}
    To compute what the error bound may be for a given value of the number of iterations, we can find that simply as
    $$\epsilon \simeq \frac{m}{2^k}$$
\end{remark}

Now for Block Graphs, recall that Kruskal's minimal spanning tree algorithm is of the order $O(m log(m))$ and if Dijkstra's longest path algorithm for trees is chosen to generate the longest path of the minimal spanning tree, then it is of the order $O(n)$.

\begin{corollary}
    Let $\mathcal{L}$ be the length of the longest chain. Then, the asymptotic complexity of Algorithm 6, to generate the longest chain of a Block graph, is
    $$Generate\_\mathcal{C}_l \in O\left(n + mlog(m) + \mathcal{L}\right)$$
\end{corollary}
\begin{proof}
    Since Kruskal's Algorithm is of the order $O(mlog(m))$ and with a choice of Dijkstra's algorithm to generate the longest path of a tree, we have an addition of $O(n)$. Finally, the number of operations to construct the chain is equal to the size of the chain, which is of the order $O(\mathcal{L})$. Therefore, we find the order to be
    $O(n + mlog(m) + \mathcal{L})$
\end{proof}
With the understanding that we may choose a different algorithm to generate the longest path, the computational complexity of Algorithm 6 is dependant on that choice as we will see when computing the asymptotic complexity of Algorithm 10. Since we know that Algorithm 7 depends on Algorithm 6, we find its complexity to be as follows.
\begin{corollary}
    Let $\mathcal{L}$ be the length of the longest chain. Then, the asymptotic complexity of Algorithm 7, to find the length of the longest path of Block Graphs, is
    $$\mathscr{LP} \in O\left(n + mlog(m) + \mathcal{L}\right)$$
\end{corollary}
\begin{proof}$ $\newline
    With the same assumptions as for Algorithm 6, we use the $Generate\_\mathcal{C}_l$ procedure and have a straightforward sum over the orders of each block of the longest chain, giving us the complexity $O(n + mlog(m) + \mathcal{L} + \mathcal{L}) = O(n + mlog(m) + \mathcal{L})$
\end{proof}
Now moving on to path-generating algorithms, we have our algebraic procedures based on Corollary 1.1 and Corollary 4.1 once again. We find Algorithm 8 and Algorithm 9 to have the same structure, hence we find
\begin{corollary}
    Assume there are $x$ longest-paths and the length of the longest path is $\ell$. Then the asymptotic complexity of Algorithm 8 and 9, to generate all possible longest paths in Trees and Directed Acyclic Graphs respectively, to be
    $$Generate\_\mathfrak{D}, Generate\_\mathfrak{D}_{_{dag}} \in O\left(x\cdot \delta n^2 log_2^2(m) \cdot \ell\right)$$
\end{corollary}
\begin{proof}$ $\newline
    The structure of the algorithm is to first find end-points of a longest path, then by using Corollary 1.1, work backwards to generate the rest of the vertices for the entire path and repeat for all possible longest paths. To check if there are end-points of a longest path, we perform $\frac{n(n+1)}{2}$ operations, once end-points have been identified, we perform a booleanized exponent for the entire length of the longest path to generate the rest of the path, this is repeated $x$ times. This requires $\ell$ exponentiation operations, which means we have an additional $\delta n^2 log_2^2(m) \cdot \ell$ operations for every identified pair of end-points. Giving us 
    $$O\left(\frac{n(n+1)}{2} + x\cdot \delta n^2 log_2^2(m) \cdot \ell\right) = O\left(x\cdot \delta n^2 log_2^2(m) \cdot \ell\right)$$
\end{proof}

As previously discussed, Algorithm 10 is somewhat of a generalization of Algorithm 6, replacing Dijkstra's longest path algorithm on Trees with our algorithm, $Generate\_\mathfrak{D}$ in Algorithm 8, to generate all possible longest paths of the minimal spanning tree of the Block Graph. Hence, we see the required number of operations to be
\begin{corollary}
    Let $\mathcal{L}$ be the size of the longest chain, and $x$ be the number of longest chains. Then, the asymptotic complexity of Algorithm 10, to generate every  longest chain of a Block Graph, is
    $$Generate\_all\_\mathcal{C}_l \in O\left(x\cdot (\delta n^2 log_2^2(m) + mlog(m))\right)$$
\end{corollary}
\begin{proof}$ $\newline
    First we begin with Kruskal's minimal spanning tree algorithm which is of the order $O(mlog(m))$. Then, we perform Algorithm 8 on the minimal spanning tree. As per Corollary 4.6, it is of the order $O(x\cdot \delta n^2 log_2^2(m) \cdot \ell)$. Then, we perform Algorithm 6 on every longest path generated from the minimal spanning tree, giving us an additional $x \cdot (n + mlog(m) + \mathcal{L})$. Given that the length of the longest path is equal to the size of the longest chain, we have $\ell = \mathcal{L}$. Hence, we find the order to be
    $$O\left(mlog(m) + x\cdot \delta n^2 log_2^2(m) \cdot \mathcal{L} + x \cdot (n + mlog(m) + \mathcal{L})\right) = O\left(x\cdot (\delta n^2 log_2^2(m) + mlog(m))\right)$$
\end{proof}
For our final algorithm, we must note that it is not possible to perform it in polynomial time, due to the fact that the number of possible longest paths increases exponentially by the order of each block. Furthermore, the asymptotic complexity will be computed assuming an upper bound so that every block's order is of the maximum size, in particular $\omega(\Gamma)$. Certainly, the only way to reach this upper bound is with Uniform Block Graphs as a special case.
\begin{corollary}
    Let $\mathcal{L}$ be the length of the longest chain, and let $x$ be the number of longest chains. Then the asymptotic complexity of Algorithm 11, to generate every longest path of a Block Graph, is
    $$Generate\_all\_\mathscr{LP} \in O\left(x\cdot (\delta n^2 log_2^2(m) + mlog(m) + \mathcal{L}\cdot\omega(\Gamma)!)\right)$$
\end{corollary}
\begin{proof}$ $\newline
    Firstly, the algorithm begins with the application of Algorithm 10, which is found to be of the order $O(x\cdot (\delta n^2 log_2^2(m) + mlog(m)))$ as per Corollary 4.7. Next, we iterate over every chain giving us a factor of $x$ for the following orders. For every chain, we iterate over every vertex of the first block. Assuming an upper bound where every block is of the order $\omega(\Gamma)$, we have an addition $\omega(\Gamma)$ term, then we append every vertex of every block of the chain. With $\mathcal{L}$ blocks and $\omega(\Gamma)$ vertices in every block, we add an overall factor of $x(\omega(\Gamma) + \mathcal{L}\cdot (\omega(\Gamma)))$. Finally, we perform permutations on non-vertex separator elements which provides us with $x \cdot \left( 2(\omega(\Gamma) - 1)! + (\mathcal{L}-2)\cdot(\omega(\Gamma) - 2)!\right)$ operations. We take all these factors together to give us an order of 
    \begin{align*}
        &O\left(x\cdot (\delta n^2 log_2^2(m) + mlog(m)) + x(\omega(\Gamma) + \mathcal{L}\cdot (\omega(\Gamma))) + x \cdot\mathcal{L}\cdot\omega(\Gamma)!\right)\\
        & \hspace{8cm} =  O\left(x\cdot (\delta n^2 log_2^2(m) + mlog(m) + \mathcal{L}\cdot\omega(\Gamma)!)\right)
    \end{align*}
\end{proof}

This completes the analysis of asymptotic complexity of the algorithms provided. 
\section{Concluding Remarks}
Overall, we have found exact algebraic solutions to find the longest path of tree-like classes of graphs. This gave rise to polynomial time algorithms that compute the length of the longest path. Through proving correctness of the algorithms, we found new algorithms to generate every possible longest path of Trees, Directed Acyclic Graphs, Uniform Block graphs, and Block Graphs which is not the case with existing pure algorithmic solutions. The algebraic approach presents itself as a promising method of efficiently solving \textbf{LPP} on classes of graphs for which a polynomial solution does not exist yet. It is possible to use this technique to attempt to find an efficient solution to \textbf{LPP} for a new class of graphs with an immediate and straightforward method to guarantee correctness. 

In general, algebraic algorithms are guaranteed to be correct given a proof of the algebraic condition it is based on, as it is what leads to the discovery of the algorithm. Hence, by defining operations which are computable in efficient time, any polynomial sequence of operations converging to the final algebraic condition would be guaranteed to solve the problem in efficient time. While not as critical as time complexity, given that a simple matrix datastructure is used to model the algebraic objects, space complexity is also in polynomial order. Besides, unlike other methods to solve \textbf{LPP}, the algebraic operations lead to new methods to generating the longest path itself. Not only a single longest path, nor only every path of the longest length, but also any path of any length. Therefore, suggesting the ability to not only solve \textbf{HPP} in polynomial time for a class of graphs, but also generate the specific hamiltonian path to traverse the graph.
\\\\
\textbf{Acknowledgements.} I would like to express my immense gratitude to Dr. Ayman Badawi for valuable discussion, consultation, mentorship, and deliberation through this work.
\\\\
\textbf{Data availability.} No new data has been generated or analyzed in this paper.
\\\\
\textbf{Conflict of interest.} The author does not have any known conflict of interest.

\end{document}